%% file: ms.tex
\def\year{2019}\relax
\pdfoutput=1
\documentclass[letterpaper]{article} %DO NOT CHANGE THIS
\usepackage{aaai19}  %Required
\usepackage{times}  %Required
\usepackage{helvet}  %Required
\usepackage{courier}  %Required
\usepackage{url}  %Required
\usepackage{graphicx}  %Required
\usepackage{amsfonts}
\usepackage{natbib}
\usepackage{nicefrac}
\usepackage{microtype}
\usepackage{amsmath}
\usepackage{tikz}
\usepackage{graphicx}
\usetikzlibrary{fit,arrows,calc,positioning,shapes.geometric}
\tikzset{cross/.style={path picture={
    \draw[black]
  (path picture bounding box.south east) -- (path picture bounding box.north west) (path picture bounding box.south west) -- (path picture bounding box.north east);
  }}}

\usepackage{mathtools}
\usepackage{amsthm}
\usepackage{amssymb}
\usepackage{paralist}
\input{latexdefs/def}

\DeclareMathSizes{11}{7}{7}{7}

\newtheorem{observation}{Observation}

\newcommand{\subt}[1]{{\triangle}_{#1}}
\newcommand{\tinfosetvalue}[1]{{V}^t_{\subt{#1}}}
\newcommand{\xhat}{{\hat x}}
\newcommand{\xhatvalue}[2]{{\hat V}^t_{\subt{#1}}(#2)}
\newcommand{\childinfosets}[1]{\mathcal{C}_{#1}}
\newcommand{\cumulativeregret}[2]{R^{#2}_{\subt{#1}}}
\newcommand{\laminarregret}[2]{{\hat R}^{#2}_{#1}}
\newcommand{\averageregret}[2]{\bar{R}^{#2}_{\subt{#1}}}

\newcommand{\losshat}[2]{\hat{\ell}^{#2}_{#1}}
\newcommand{\convsets}{\mathcal{J}}

\DeclareRobustCommand{\obspoint}{\tikz{\node[scale=.85, draw, cross, circle, inner sep=1mm] {}}}
\DeclareRobustCommand{\enddecproc}{\tikz{\node[draw, circle, regular polygon, regular polygon sides=6, shape border rotate=180, inner sep=.7mm, fill=black] {}}}

\usepackage{enumitem}
\usepackage{algorithm}
\usepackage{algorithmic}
\usepackage{thm-restate}
\usetikzlibrary{arrows,arrows.meta}
\DeclareRobustCommand{\computeloss}{\tikz[scale=.25]{\draw[thick] (0, 0) -- (1, 0) -- (1, 1) -- (0, 1) -- (0, 0) -- (1.0, 0); \draw[thick] (0, 0) -- (1, 1);}}

\begin{document}

 \title{Online Convex Optimization for Sequential Decision Processes and
   Extensive-Form Games}
\author{Gabriele Farina \and Christian Kroer \and Tuomas Sandholm\\Computer Science Department\\Carnegie Mellon University\\\texttt{\{gfarina,ckroer,sandholm\}@cs.cmu.edu}}
\maketitle
\begin{abstract}
  \input{text/abstract}
\end{abstract}

\input{text/intro}
\input{text/regret_minimization}

\input{text/sequential_action_spaces}
\input{text/dilated_tree_regret_minimization}
\input{text/lrd_generalized}
%\input{text/laminar_regret_minimization}
\input{text/cfr_special_case}
\input{text/sequence_form}
\section{Application Domains}
\input{text/treeplex_special_case}

\input{text/efg_special_case}

\input{text/quantal_special_case}

\input{text/perturbed_efg_special_case}
\input{text/cfr_plus}
\input{text/experiments}
\input{text/conclusion}

\section*{Acknowledgments}
This material is based on work supported by the National Science Foundation under grants IIS-1718457, IIS-1617590, and CCF-1733556, and the ARO under award W911NF-17-1-0082. Christian Kroer is supported by a Facebook Fellowship.

\bibliography{dairefs/dairefs}
\bibliographystyle{aaai}

\clearpage
\appendix
\input{text/omitted_proofs}

\input{text/appendix_experiments}

\end{document}

%% file: latexdefs/def.tex
\usepackage{etoolbox}
\ifcsmacro{assumption}{}{
	
}
\ifcsmacro{corollary}{}{
	\newtheorem{corollary}{Corollary}[]
}
\ifcsmacro{definition}{}{
	
}
\ifcsmacro{example}{}{
	
}
\ifcsmacro{fact}{}{
	
}
\ifcsmacro{informal_theorem}{}{
	
}
\ifcsmacro{lemma}{}{
	
}
\ifcsmacro{proposition}{}{
	
}
\ifcsmacro{remark}{}{
	
}
\ifcsmacro{theorem}{}{
	\newtheorem{theorem}{Theorem}[]
}

\usepackage{mathtools}
\mathtoolsset{centercolon}
\newcommand{\defeq}{\mathrel{:\mkern-0.25mu=}}

\newcommand{\R}{\ensuremath{\mathbb{R}}}

\newcommand{\cfrp}{CFR$^+$}
\newcommand{\rmp}{RM$^+$}

\newcommand{\be}{\begin{eqnarray}}
\newcommand{\ee}[1]{\label{#1}\end{eqnarray}}

	\newcommand{\ese}{\end{eqnarray*}}
	\newcommand{\bse}{\begin{eqnarray*}}
	\def\beq{\begin{equation}}
	\def\eeq{\end{equation}}

	\def\fnote#1{\footnote}

	\def\*{{{\LARGE\bf $^*$}}}

	\def\R{{\mathbb{R}}}

%% file: text/abstract.tex
Regret minimization is a powerful tool for solving large-scale
extensive-form games. State-of-the-art methods rely on minimizing regret locally
at each decision point. In this work we derive a new framework for regret
minimization on sequential decision problems and extensive-form games with
general compact convex sets at each decision point and general convex losses, as
opposed to prior work which has been for simplex decision points and linear
losses. We call our framework \emph{laminar regret decomposition}. It 
generalizes the CFR algorithm to this more general setting. Furthermore, our
framework enables a new proof of CFR even in the known setting, which is
derived from a perspective of decomposing polytope regret, thereby leading to an
arguably simpler interpretation of the algorithm. Our generalization to convex
compact sets and convex losses allows us to develop new algorithms for several
problems: regularized sequential decision making, regularized Nash equilibria in
extensive-form games, and computing approximate extensive-form perfect
equilibria. Our generalization also leads to the first regret-minimization
algorithm for computing reduced-normal-form quantal response equilibria
based on minimizing local regrets. Experiments show that our framework
leads to algorithms that scale at a rate comparable to the fastest variants of
counterfactual regret minimization for computing Nash equilibrium, and therefore
our approach leads to the first algorithm for computing quantal response
equilibria in extremely large games. Finally we show that our framework enables a
new kind of scalable opponent exploitation approach.
%%% Local Variables:
%%% mode: latex
%%% TeX-master: "../aaai_2019/lrd_aaai_2019"
%%% End:

%% file: text/intro.tex
\section{Introduction}

\emph{Counterfactual regret minimization (CFR)}~\citep{Zinkevich07:Regret}, and the newest variant \emph{{\cfrp}}~\citep{Tammelin15:Solving}, have been a
central component in several recent milestones in solving imperfect-information \emph{extensive-form games (EFGs)}. \citet{Bowling15:Heads} used {\cfrp} to near-optimally solve heads-up limit Texas hold'em.
\citet{Brown17:Superhuman} and \citet{Moravvcik17:DeepStack} used CFR variants, along with
other scalability techniques, to create AIs that beat professional poker
players at the larger game of heads-up no-limit Texas hold'em. 

We can view the CFR approach more
generally as a methodology for setting up regret minimization for
sequential decision problems (whether single- or multi-agent), where each
decision point requires selecting either an action or a point from the
probability distribution over actions. The crux of CFR is counterfactual regret,
which leads to a definition of regret local to each decision point. CFR can then
be viewed as the observation, and proof, that bounds on counterfactual regret,
which can be minimized locally, lead to bounds on the overall regret. To minimize local regret, the framework relies on regret
minimizers that operate on a simplex (typically of probabilities over the available actions), such as \emph{regret matching} (RM)~\citep{Blackwell56:analog} or the
newer variant \emph{regret matching$^+$} (\rmp)~\citep{Tammelin15:Solving}.

In this paper we consider the more general problem of how to minimize regret
over a sequential decision-making (SDM) polytope, where we allow arbitrary
compact convex subsets of simplexes at each decision point (as opposed to only
simplexes in CFR), and general convex loss functions (as opposed to only linear
losses in CFR). This allows us to model a form of online convex optimization
over SDM polytopes. We derive a decomposition of the polytope regret into local
regret at each decision point. This allows us to minimize regret locally as with
CFR, but for general compact convex decision points and convex losses. We call
our decomposition \emph{laminar regret decomposition (LRD)}. We call our overall
framework for convex losses and compact convex decision points \emph{laminar
  regret minimization (LRM)}. As a special case, our framework provides an
alternate view of why CFR works---one that may be more intuitive for those
with a background in online convex optimization.
% From the perspective of CFR we
% derive several generalizations: to general bounded convex sets at each decision
% point (as opposed to exactly the simplex), and to general convex losses at each
% decision point (as opposed to only linear losses).

Our generalization to general compact convex sets (we restrict our attention to
convex subsets of simplexes, but this is without loss of generality due to the
convexity-preserving properties of affine transformations) allows us to model
entities such as $\epsilon$-perturbed
simplexes~\citep{Farina17:Extensive,Farina17:Regret,Kroer17:Smoothing}, and thus
yields new algorithms for computing approximate equilibrium refinements for
EFGs.

General convex losses in SDM and EFG contexts have, to the best of our
knowledge, not been considered before. This generalization enables fast
algorithms for many new settings. One is to compute
regularized zero-sum equilibria. If we apply a convex regularization function at
each simplex, we can apply our framework to solve the resulting game. For
the negative entropy regularizer this is equivalent to the dilated entropy
distance function used for solving EFGs with first-order
methods~\citep{Hoda10:Smoothing,Kroer15:Faster,Kroer17:Theoretical}.
\citet{Ling18:What} show that dilated-entropy-regularized EFGs are equivalent to
quantal response equilibria (QRE) in the corresponding reduced normal-form game.
Thus our result yields the first regret-minimization algorithm for computing
reduced-normal-form quantal response equilibria in EFGs.

Our experiments show that QREs and
$\ell_2$-regularized equilibria can be computed at a rate that is competitive
with that of {\cfrp} for computing Nash equilibria, and substantially faster in
some cases. This shows that our approach can be used to compute regularized
equilibria in extremely large games such as real poker games. We go on to
show that our framework also enables a new kind of opponent-exploitation
approach for extremely large games, by adding a convex regularizer that penalizes
the exploiter for being far away from a pre-computed Nash equilibrium, and thus potentially exploitable herself.

%%\subsection{Overview of our framework}
%\textbf{Our framework.} We derive a sequence of general statements relating to regret on
%SDM polytopes.
%%In this section we give an overview
%%of how these statements can be put together in order to construct algorithms for
%%solving SDM problems or EFGs.
%Faced with a SDM problem, we would like to construct an algorithm for
%regret minimization over the polytope. Our central result is the
%laminar regret decomposition (LRD) in Theorem~\ref{thm:lrd}, which shows that
%the regret over the whole SDM polytope can be decomposed into regret at each
%decision point, and Theorem~\ref{thm:lrd_rm} which shows that overall regret can
%be minimized by minimizing these local regrets. Thus we combine LRD with a set
%of local regret minimizers, one for each decision point. The local regret
%minimizers are each given the modified losses shown in \eqref{eq:losshat}. This
%gives a regret-minimization algorithm for SDM.
%
%We build algorithms for EFGs on top of the SDM setup: each player runs their
%own SDM regret minimizer for their decision space. The loss that each player
%faces at iteration $t$ is then the negative payoff vector associated with the
%strategy of the opposing player at the previous iteration. As explained in
%Theorem~\ref{thm:folk theorem} we can then construct an $\epsilon$-Nash
%equilibrium by using the average regret-minimizer strategy of each player, where
%the averaging is performed in the sequence form.

%%% Local Variables:
%%% mode: latex
%%% TeX-master: "../aaai_2019/lrd_aaai_2019"
%%% End:

% LocalWords:  LRD SDM

%% file: text/regret_minimization.tex
\section{Regret Minimization}
We work the online learning framework called \emph{online convex optimization}~\citep{Zinkevich03:Online}.
%In this section, we will briefly touch on the important ideas used in this paper.
In this setting, a decision maker repeatedly plays against an unknown environment by making a sequence of decisions $x^1, x^2, \dots$. As customary, we assume that the set $X \subseteq \mathbb{R}^n$ of all possible decisions for the decision maker is convex and compact. The outcome of each decision $x^t$ is evaluated as $\ell^t(x^t)$, where $\ell^t$ is a convex function \emph{unknown} to the decision maker until the decision is made. Abstractly, a \emph{regret minimizer} is a device that supports two operations:
\begin{itemize}[nolistsep,itemsep=0mm]
  \item it gives a \emph{recommendation} for the next decision $x^{t+1} \in X$;
  \item it receives/observes the convex loss function $\ell^t$ used to ``evaluate'' decision $x^t$.
\end{itemize}
The learning is \emph{online} in the sense that the decision maker/regret minimizer's next decision, $x^{t+1}$, is based only on the previous decisions $x^1, \dots, x^t$ and corresponding loss observations $\ell^1, \dots, \ell^t$.

In this paper, we adopt (external) \emph{regret} as a way to evaluate the quality of the regret minimizer. Formally, the \emph{cumulative regret} at time $T$ is defined as

\vspace{-3mm}
\[
  R^T \defeq \sum_{t=1}^T \ell^t(x^t) - \min_{\hat x \in X} \sum_{t=1}^T \ell^t(\hat x),
\]
\vspace{-3mm}

\noindent It measures the difference between the loss cumulated by the sequence of decisions $x^1, \dots, x^T$ and the loss that would have been cumulated by playing the best time-independent decision $\hat x$ in hindsight. A desirable property of a regret minimizer is
\emph{Hannan consistency}: the average regret approaches zero, that is, $R^T$ grows at a \emph{sublinear} rate in $T$.

We now review a particular very general regret-minimization algorithm:
\emph{online mirror descent} (OMD). The generality of OMD arises because it
performs updates in the dual space, where the duality is given by a
\emph{mirror map} $\Phi$, a strongly-convex differentiable function on $X$ which
defines a vector field in which gradient updates are performed.
At each time step OMD performs the following update:

\vspace{-3mm}
\[
  \nabla \Phi (y_{t+1}) = \nabla \Phi (y_t) - \eta \nabla \ell^t(x_t),
\]
\vspace{-3mm}

and then recommends the point

\vspace{-3mm}
\[
  x_{t+1} = \arg\min_{x \in X} \Phi(x) - \langle \nabla\Phi(y_{t+1}), x \rangle.
\]
\vspace{-3mm}

If OMD is initialized with $\nabla\Phi(y_1)=0$ and $x_1$ as the corresponding
minimizer, it satisfies the regret bound

\vspace{-3mm}
\[
  R^T \leq \max_{u,v\in X} \{\Phi(u)-\Phi(v)\} + \eta \sum_{t=1}^T
  \|\nabla\ell^t(x_t) \|^2_*,
\]
\vspace{-3mm}

\noindent where $\|\cdot\|_*$ is the dual norm with respect to which $\Phi$ is strongly convex. OMD is very general in the sense that we can choose $\Phi$ and the norm for
measuring strong convexity so that it fits the problem at hand. For example,
this allows only a logarithmic dependence on the dimension of $X$ when
$X=\Delta_n$ and $\Phi$ is the negative entropy. By specific choices of $\eta$
and $\Phi$ it is possible to show that this algorithm generalizes online
variants of gradient descent, exponential weights, and regularized follow-the-leader~\citep{Zinkevich03:Online,Hazan10:Extracting,Hazan16:Introduction}. The regret generally grows
at a rate of $T^{-1/2}$ for these algorithms.

We could also run OMD with $X$ being the entire SDM polytope. For example, we
could do that by applying the \emph{distance generating function (DGF)} of \citet{Kroer17:Theoretical}. However,
decomposition into local regret minimization at each decision point has been
dramatically more effective in practice, possibly because this allows better leveraging
of the structure of the problem.

% \todo{Talk about lower bounds}

\textbf{Linear losses and games.} Regret minimization methods for normal-form
and extensive-form games usually involve minimizing the regret induced by linear
loss functions. When the domain at each decision point $X_j$ is the
$n_j$-dimensional simplex $\Delta_{n_j}$, the two most successful
regret-minimizers in practice have been \emph{regret matching}~\citep{Blackwell56:analog} and
\emph{regret matching$^+$}~\citep{Tammelin15:Solving}. These regret minimizers
also have regret that grows at a rate $T^{-1/2}$ as with OMD, but they have a
worse dependence on the dimension $n_j$. Nonetheless, they seem to perform
better in practice when coupled with CFR.

%% file: text/sequential_action_spaces.tex
\section{Sequential Decision Making}
It turns out that the results of this paper can be proven in a general setting which we call a sequential decision making. At each stage, the
agent chooses a point in a simplex (or a subset of it). The chosen point incurs a
convex loss and defines a probability distribution over the actions of the
simplex. An action is sampled according to the chosen distribution, and the
agent then arrives at a next decision point, potentially randomly selected out
of several candidates. The reason the agent chooses points in the convex hull of
actions, rather than simply an action, is that this gives us greater flexibility
in representing decision points where agents wish to randomize over actions.
This is the case for example in game-theoretic equilibria or when solving the
decision-making problem with an iterative optimization algorithm.
%THIS SENTENCE WAS CUT OFF. WHAT DID CHRISTIAN MEAN TO SAY HERE: ???
%It can also occur

Formally, we assume that we have a set of decision points $\convsets$. Each
decision point $j\in \convsets$ has a set of actions $A_j$ of size $n_j$. The
decision space at each decision point $j$ is represented by a convex set $X_j
\subseteq
\Delta_{n_j}$. % that acts as a root of the decision tree underneath it.
A point $x_j \in X_j$ represents a probability distribution over $A_j$. When a
point $x_j$ is chosen, an action is sampled randomly according to $x_j$. Given a
specific action at $j$, the set of possible decision points that the agent may
next face is denoted by $\childinfosets{j,a}$. It can be an empty set if no
more actions are taken after $j,a$. We assume that the decision points form a
tree, that is, $\childinfosets{j,a} \cap \childinfosets{j',a'} = \emptyset$ for all
other convex sets and action choices $j',a'$. This condition is equivalent to
the perfect-recall assumption in extensive-form games, and to conditioning on
the full sequence of actions and observations in a finite-horizon
partially-observable decision process.
%We let $Q_j$ be the tree-like structure
%of decision points starting with $X_j$ as the root. We let $Q=Q_r$ refer to the
%set of all decision points starting at the root convex set $X_r$.
In our
definition, the decision space starts with a root decision point, whereas in
practice multiple root decision points may be needed, for example in order to
model different starting hands in card games. Multiple root decision points can
be modeled in our framework by having a dummy root decision point with only a single action.

The set of possible next decision points after choosing action $a\in A_j$ at
$X_j$, denoted $\childinfosets{j,a}$, can be thought of as representing the
different decision points that an agent may face after taking
%an
action $a$
and then making an observation on which she can condition her next action choice.
For example, in a card game an action may be to raise (that is, put money into the
pot), and an observation could
% would then
be the set of actions taken by the other
players, as well as any cards dealt out, until the agent acts again. Each
specific observation of actions and cards then corresponds to a specific
decision point in $\childinfosets{j,a}$.

We will relate the regret over the whole decision space to regret at
subtrees in the decision space and individual convex sets. In order to do that
we need ways to refer to each of these structures. Given a strategy $x$,
 $x_j$ is the (sub)vector belonging to the decision space $X_j$ at decision point $j$. Similarly, $x_{j, a}$
is the scalar associated with action $a \in
A_j$ at decision point $j$, and in typical applications it is the probability of choosing action $a$ at decision point $j$. Subscript $\subt{j}$ denotes the portion of $x$ containing the decision
variables for decision point $j$ and all its descendants. Finally, $x$
refers to the vector for the whole treeplex, which corresponds to subscript
$\subt{r}$ where $r$ is the root of the tree.

As an illustration, consider the game of Kuhn poker~\citep{Kuhn50:Simplified}.
Kuhn poker consists of a three-card deck: king, queen, and jack. Each player is
dealt one of the three cards and a  single round of betting occurs. A complete
game description is given in the appendix.
The action space for the first player is shown in Figure~\ref{fig:kuhn treeplex
  player1}. For instance, we have: $\convsets = \{0,1,2,3,4,5,6\}$; $n_0 = 1,
X_0 = \Delta_1 = \{1\}$; $n_j = 2, X_j = \Delta_2$ for all $j\in\convsets
\setminus \{0\}$; $A_0 = \{\text{start}\}$, $A_1 = A_2 = A_3 = \{\text{check},
\text{raise}\}$, $A_4 = A_5 = A_6 = \{\text{fold}, \text{call}\}$;
$\childinfosets{0,\text{start}} = \{1, 2, 3\}$, $\childinfosets{1, \text{raise}}
= \emptyset$, $\childinfosets{3,\text{check}} = \{6\}$; $X_{\subt{1}}= X_1\times
X_4$, $X_{\subt{4}}=X_4$; $x_1 = [x_{1,\text{check}},x_{1,\text{raise}}]$; $x_{\subt{1}} =
[x_1;x_4]$, etc.

\begin{figure}[!h]
  \vspace{-3mm}
  \centering\includegraphics[width=.7\linewidth]{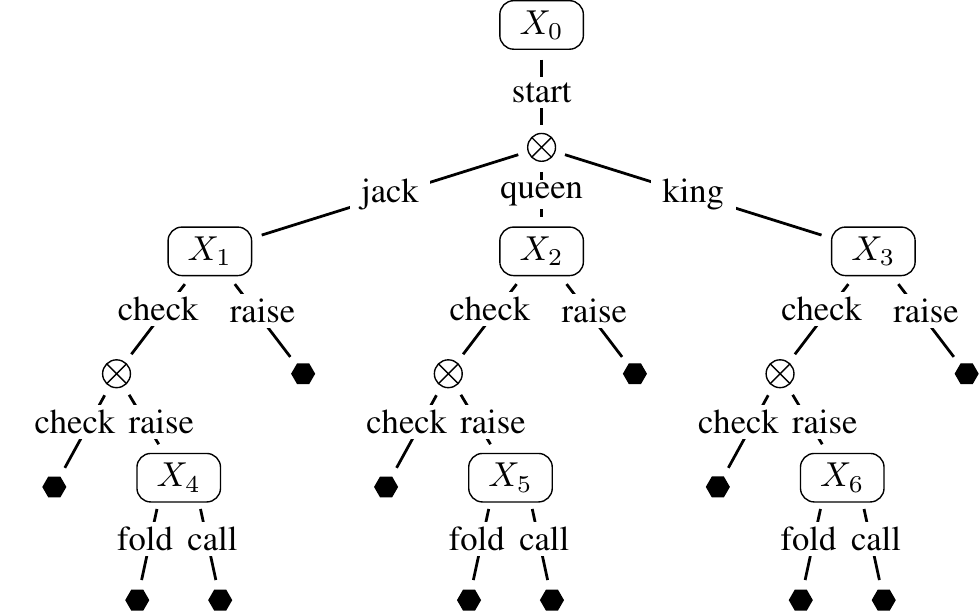}
  \vspace{-3mm}
  \caption{Sequential action space for the first player in the game of Kuhn poker. \raisebox{-.2mm}{\obspoint{}} denotes an observation point; \enddecproc{} represents the end of the decision process.  \vspace{-4mm}}
  \label{fig:kuhn treeplex player1}
\end{figure}

% TODO:
% - \pi_j
% Our space of decisions captures many sequential decision-making problems.
In addition to games, our model captures, for example, POMDPs and MDPs where we condition on the entire history of
observations and actions.
%Finite-horizon
% Later in the paper we will discuss several applications in detail.

% \subsection{Notation}
\vspace{-.3cm}

%%% Local Variables:
%%% mode: latex
%%% TeX-master: "../aaai_2019/lrd_aaai_2019"
%%% End:

%% file: text/dilated_tree_regret_minimization.tex
\section{Regret in Sequential Decision Making}

We assume that we are playing a sequence of $T$ iterations of a sequential
decision process. At each iteration $t$ we choose a strategy $x\in X$ and are
then given a loss function of the form
{\footnotesize
\begin{align}
    \vspace{-1mm}
  \ell^t(x) \defeq \sum_{j \in \convsets} \pi_j(x) \ell^t_j(x_j),
  \label{eq:loss}\vspace{-1mm}
\end{align}}
where $\ell^t_j: X_j \rightarrow \R$ is a convex function for each $j\in
\convsets$.
We coin loss functions of this form \emph{separable}, and they will play an important role in our results.
Our goal is to compute a new strategy vector $x^t$ such that the
regret across all $T$ iterations is as low as possible against any sequence of
loss functions.

We now summarize definitions for the value and regret associated
with convex sets and strategies.
First we have the value of convex set $j$ at iteration $t$ when following
 strategy $\xhat$:
{\footnotesize
\[\vspace{-1mm}
  \xhatvalue{j}{\xhat_{\subt{j}}} \defeq \ell^t_j(\xhat^t_j) + \sum_{a\in A_j}
  \sum_{j' \in \childinfosets{j,a}} \hat{x}_{j,a} \xhatvalue{j'}{\xhat_{\subt{j'}}}.
  \vspace{-1mm}
\]}
This definition denotes the utility associated with starting at convex set $X_j$
rather than at the root. Thus we have exchanged the term $\pi_j(x)$ with one for
$\ell^t_j$ and with $\xhat_{j,a}$ for $\tinfosetvalue{j'}$; this allows us to
write the value as a recurrence.
We will be particularly interested in the value of $x^t$, which we denote
$
  \tinfosetvalue{j} \defeq \hat{V}_{\subt{j}}^t(x^t_{\subt{j}}).
$
% \[
%   \xhatvalue{j} = \sum_{a \in A_j} \hat{x}^t_{j,a} g^t_{j,a} + \sum_{a\in A_j}
%   \sum_{j' \in \childinfosets{j,a}} \hat{x}^t_{j,a} \xhatvalue{j'}
% \]

Now we can define the cumulative regret at convex set $j$ across all $T$
iterations as

\vspace{-2mm}
{\footnotesize
\[
  \cumulativeregret{j}{T} \defeq  \sum_{t=1}^T \tinfosetvalue{j} - \min_{\xhat_{\subt{j}}}  \sum_{t=1}^T.
    \xhatvalue{j}{\xhat_{\subt{j}}},
\]
}
\vspace{-2mm}

This can equivalently be stated as

\vspace{-4mm}
{\footnotesize
\begin{align}
   \min_{\xhat_{\subt{j}}}  \sum_{t=1}^T
    \xhatvalue{j}{\xhat_{\subt{j}}} =   \sum_{t=1}^T \tinfosetvalue{j} - \cumulativeregret{j}{T}.
  \label{eq:regret_j_rewrite}
\end{align}}
  \vspace{-2mm}
  
Finally, \emph{average regret} is $\averageregret{j}{T}= \frac{1}{T}\cumulativeregret{j}{T}$.

%%% Local Variables:
%%% mode: latex
%%% TeX-master: "../aaai_2019/lrd_aaai_2019"
%%% End:

%% file: text/lrd_generalized.tex
\section{Laminar Regret Decomposition}

We now define a new parameterized class of loss functions for each subtree $X_j$
which we will show can be used to minimize regret over $X$ by minimizing that
loss function independently at each convex set $X_j$. The loss function is
% add gamma weight term and make it parametric loss function. Gamma should
% multiply the partial term

\vspace{-4mm}
\begin{align}
  \losshat{j}{t}(x_j)\! \defeq\! \ell^t_j(x_j)
  +\! \sum_{a\in A_j} \sum_{j' \in \childinfosets{j,a}}
  x_{j,a} \tinfosetvalue{j'} .
  \label{eq:losshat}
\end{align}
\vspace{-4mm}

It is convex since $\ell^t_j$  is convex by hypothesis
and we are only adding a linear term to it. Strict convexity is also preserved,
and for strongly convex losses, the strong convexity parameter remains unchanged.

We now prove that the regret at information set $j$ decomposes into regret terms
depending on $\hat{\ell}^t_j$ and a sum over the regret at child convex sets:
% theorem should be for every \gamma \in [0,1]
% should get gamma times residual, and then \leq arises from the (1-gamma)
% part that we get from taking a max. Say that it's equality when gamma is 1.
\begin{theorem}  \label{thm:lrd}
  The cumulative regret at a decision point $j$ can be decomposed as

  \vspace{-5mm}
  {
    \footnotesize
  \begin{align*}
    &\cumulativeregret{j}{T}
    = \sum_{t=1}^T  \losshat{j}{t}(x^t) - \min_{\xhat_j \in X_j} \bigg\{\!
    \sum_{t=1}^T \losshat{j}{t}(\xhat_j)
    - \sum_{a \in A_j} \sum_{j' \in \childinfosets{j,a}} \!\!\xhat_{j,a}\cumulativeregret{j'}{T}
    \bigg\}
  \end{align*}
  }
    \vspace{-4mm}

\end{theorem}
\vspace{-3mm}
\begin{proof}
  By definition, the cumulative regret $\cumulativeregret{j}{T}$ at time $T$ for decision point $j$ is:

  \vspace{-4mm}
  {\footnotesize
  \begin{align}
    &\sum_{t=1}^T \tinfosetvalue{j}
     -  \min_{\xhat_{\subt{j}}} \bigg\{  \sum_{t=1}^T
        \ell^t_j(\xhat_j) +  \sum_{t=1}^T\sum_{a\in A_j}\sum_{j' \in \childinfosets{j,a}} \xhat_{j,a}\xhatvalue{j'}{\xhat_{\subt{j'
     }}}
                          \bigg\} \nonumber\\[-3mm]
    &= \sum_{t=1}^T \tinfosetvalue{j}
      -  \min_{\xhat_j \in X_j} \bigg\{  \sum_{t=1}^T
        \ell^t_j(\xhat_j)\nonumber\\[-3mm]
     &\hspace{1.5cm} + \sum_{a\in A_j}\sum_{j' \in \childinfosets{j,a}} \hat{x}_{j,a} \min_{\xhat_{\subt{j'}}}  \sum_{t=1}^T \xhatvalue{j'}{\xhat_{\subt{j'}}}
                          \bigg\},
        \label{eq:lrd_proof_eq2}
  \end{align}
  }
  
  \vspace{-2mm}
  \noindent where the equalities follow first from expanding the definitions of
  $\cumulativeregret{j}{T}$ and $\xhatvalue{j}{\xhat_{\subt{j}}}$, and then using the fact that we can
  sequentially minimize first over choices at $j$ and then over choices for
  child information sets.

  Now we can use \eqref{eq:regret_j_rewrite} to get that \eqref{eq:lrd_proof_eq2} is equal to

  \vspace{-4mm}
  {\footnotesize
  \begin{align}
    &
    %=
%         & \sum_{t=1}^T \tinfosetvalue{j}
%         - \min_{\hat{x}_{j} \in X_j} \bigg( \sum_{t=1}^T
%         \ell^t_j(\xhat_j) \nonumber\\[-2mm]
%         &\hspace{1.0cm}+ \sum_{a\in A_j} \hat{x}_{j,a} \sum_{j' \in \childinfosets{j,a}}  \big( \sum_{t=1}^T \tinfosetvalue{j'} - \cumulativeregret{j'}{T} \big)
%                          \bigg) \nonumber\\
    = \sum_{t=1}^T \tinfosetvalue{j}
      \!-\! \min_{\hat{x}_{j} \in X_j}  \bigg( \sum_{t=1}^T \losshat{j}{t}(\xhat_j)
         \!-\! \sum_{a\in A_j} \hat{x}_{j,a}\! \sum_{j' \in \childinfosets{j,a}}\!\! \cumulativeregret{j'}{t} \bigg).
         \label{eq:lrd_proof_eq3}
  \end{align}}
  \vspace{-2mm}

  Since $ \tinfosetvalue{j}$ already depends on $\tinfosetvalue{j'}$ for each
  child decision point $j'$ we have
  $\tinfosetvalue{j} = \losshat{j}{t}(x^t_j)$,
  where the equality follows by the definition of $\losshat{j}{t}$. Substituting
  this equality in \eqref{eq:lrd_proof_eq3} yields the statement.
%  \begin{align*}
%    & \sum_{t=1}^T \losshat{j}{t}(x^t_j, \gamma)
%    +\gamma\sum_{a\in A_j}\sum_{j' \in
%    \childinfosets{j,a}} \sum_{t=1}^T x^t_{j,a} \regretdiff{j'}{T} \\[-2mm]
%    &- \min_{\hat{x}_{j} \in X_j}\bigg\{ \sum_{t=1}^T \losshat{j}{t}(\xhat_j,\gamma)
%    - (1-\gamma)\sum_{a\in A_j} \hat{x}_{j,a} \sum_{j' \in \childinfosets{j,a}} \cumulativeregret{j'}{t} \bigg\},
%  \end{align*}
%  which is exactly what we wanted to show.
\end{proof}

Theorem~\ref{thm:lrd} justifies the introduction of the concept of \emph{laminar regret} at each decision point $j\in \convsets$:

\vspace{-3mm}
{\footnotesize\[
  \laminarregret{j}{T} \defeq \sum_{t=1}^T \losshat{j}{t}(x^t_j) - \min_{\hat{x}_j\in X_j}\sum_{t=1}^T \losshat{j}{t}(\hat{x}_j).
\]}
\vspace{-2mm}

With this, we can write the cumulative subtree regret at decision point $j$ as a
sum of laminar regret at $j$ plus a recurrence term for each child decision
point. Applying this inductively gives the following theorem which tells us how
one can apply regret minimization locally on laminar regrets in order to
minimize regret in SDMs:
% We now show how the decomposition theorem (Theorem~\ref{thm:lrd}) leads to a
% regret minimization scheme for SDMs. By applying Corollary~\ref{cor:relaxed
%   decomposition} inductively we get that the cumulative regret is bounded by a
% weighted sum of laminar regrets:
\begin{restatable}{theorem}{thmlrdrm}  \label{thm:lrd_rm}
  The cumulative regret on $X$ satisfies

  \vspace{-1mm}{\footnotesize$$ R^T \le \max_{\xhat \in X}\sum_{j \in \convsets} \pi_j(\xhat) \laminarregret{j}{T}. $$}\vspace{-2mm}
\end{restatable}
\begin{corollary}
 If each individual laminar regret $\laminarregret{j}{T}$ on each of the
  convex domains $X_j$ grows sublinearly, overall regret on $X$ grows
  sublinearly.
\end{corollary}

Theorem~\ref{thm:lrd_rm} shows that overall regret can be minimized by
minimizing each laminar regret separately. In particular, this means that if we
have a regret minimizer for each decision point $j$ that can handle the
structure of the convex set $X_j$ and the convex loss from \eqref{eq:losshat},
then we can apply those regret minimizers individually at each information set,
and Theorem~\ref{thm:lrd_rm} guarantees that overall regret will be bounded by a
weighted sum over those local regrets. For example, if each local regret
minimizer has regret that grows at a particular sublinear rate, then the overall
regret is also guaranteed to grow only at that sublinear rate.
% The resulting algorithm is summarized in Algorithm~\ref{algo:lrd}.

% \begin{algorithm}
% \caption{\todo{}}\label{algo:lrd}
% \todo{Algorithm here}
% \end{algorithm}

%%% Local Variables:
%%% mode: latex
%%% TeX-master: "../aaai_2019/lrd_aaai_2019"
%%% End: 

%% file: text/cfr_special_case.tex
% CFR and its variants are the practical state-of-the-art for computing Nash
% equilibria in zero-sum EFGs. For example, CFR variants were used in the
% Libratus~\citep{Brown17:Superhuman} and DeepStack~\citep{Moravvcik17:DeepStack}
% agents, which beat professional human poker players. The CFR algorithm has also
% inspired algorithms for partially-observable deep reinforcement
% learning~\citep{Jin18:Regret}. In this section we show that a special case of
% Theorem~\ref{thm:lrd} leads to an alternative proof of the CFR result.

Our result gives an alternative proof of CFR. This is arguably simpler than
existing proofs, because we show directly why regret over a sequential
decision-making space decomposes into individual regret terms, as opposed to
bounding terms in order to fit the CFR framework.
Finally, our result also generalizes CFR to new settings: we show how CFR can be
implemented on arbitrary convex subsets of simplexes and with convex losses
rather than linear.

%%% Local Variables:
%%% mode: latex
%%% TeX-master: "../aaai_2019/lrd_aaai_2019"
%%% End:

%% file: text/sequence_form.tex
\subsection{Sequence form for sequential decision processes}
So far we have described the decision space as a product of convex sets where
the choice of each action is taken from a subset of a simplex $X_j\subseteq
\Delta_{n_j}$. This formulation has a drawback: the expected-value function for
a given strategy is not linear. Consider taking action $a$ at decision point $j$. In order to compute the expected overall contribution of that decision, its local payoff $g_{j,a}$ has to be weighted by the product of probabilities
of all actions on the path to $j$ and by $x_{j,a}$. So, the overall expected utility is nonlinear and non-convex.  We now present a well-known alternative representation of
this decision space which preserves linearity. While we will mainly be working in
the product space $X$, it will occasionally be useful to move to this equivalent
representation to preserve linearity.

The alternative formulation is called the \emph{sequence form}. In that representation, 
every convex set $j\in \convsets$ is scaled by the parent variable leading to
$j$. In other words, the sum of values at $j$ now sum to the value of the parent
variable. In this formulation, the value of a particular action then represents
the probability of playing the whole \emph{sequence} of actions from the root to
that action. This allows each term in the expected loss to be weighted only by
the sequence ending in the corresponding action. The sequence form has been used
to instantiate linear programming~\citep{Stengel96:Efficient} and first-order
methods~\citep{Hoda10:Smoothing,Kroer15:Faster,Kroer17:Theoretical} for
computing Nash equilibria of zero-sum EFGs. There is a straightforward mapping
between any $x\in X$ to its corresponding sequence form: simply assign each
sequence the product of probabilities in the sequence. Likewise, going from
sequence form to $X$ can be done by dividing each $x_{j,a}$ by the value
$x_{p_j}$ where $p_j$ is the entry in $x$ corresponding to the parent of $j$. We
let $\mu$ be a function that maps each $x\in X$ to its corresponding
sequence-form vector. For the reverse direction $\mu^{-1}$, there is 
ambiguity because $\mu$ is not injective. Nonetheless, an inverse can be
computed in linear time.

% \todo{take out this paragraph?} There is an interesting connection between
% our loss function \eqref{eq:loss} and the dilated distance functions used to
% instantiate first-order
% methods~\citep{Hoda10:Smoothing,Kroer15:Faster,Kroer17:Theoretical}: if our loss
% function is formulated in terms of the sequence-form polytope then it
% corresponds to a sum over dilated convex functions on simplexes, just as the
% dilated distance functions.

% Finally, depending on the nature of the $X_j$'s, the map $\sigma$ might not be injective: two different behavioral strategies $x, y \in Q$ might me mapped to the same $\sigma(x) = \sigma(y)$. However, it is always possible to compute an element of $\sigma^{-1}(s)$ in linear time in the dimension of $Q$.

%%% Local Variables:
%%% mode: latex
%%% TeX-master: "../aaai_2019/lrd_aaai_2019"
%%% End:

%% file: text/treeplex_special_case.tex
%\subsection{Sequential Decision Making and POMDPs}
\label{sec:sequential_decision_making}

Because we only need a finite sequential tree structure of the decision space, our
framework captures a very broad class of SDM problems. In this section we
describe how our framework can be applied to a number of prominent applications,
such as POMDPs and EFGs. In general, our framework can be applied to any SDM
problem where one or more agents are faced with a finite sequence of decisions
that form a tree, such that agents always remember all past actions. The
specific decision problem at each stage may depend on the past decisions as well
as stochasticity. The fact that we require the decision space to be tree
structured might seem limiting from the perspective of compactly representing
the decision space. However, this has successfully been dealt with in 
applications by using state- or value-estimation techniques, rather than fully
representing the original problem~\citep{Moravvcik17:DeepStack,Jin18:Regret}.

One example class of a single-agent decision problems that we can model 
is finite-horizon POMDPs where the history of states and actions
is remembered. In that case, each decision point corresponds to a specific
sequence of actions and observations made by the agent. This setting is
reminiscent of the POMDP setting considered by \citet{Jin18:Regret}. This type
of model can be used to model sequential medical treatment planning when
combined with results on imperfect-recall
abstraction~\citep{Lanctot12:No,Chen12:Tractable,Kroer16:Imperfect}, and has
potential applications in steering evolutionary
adaptation~\citep{Sandholm15:Steering,Kroer16:Sequential}. Our framework allows
more general models for such problems via our generalization to convex decision
points and convex losses; for example our framework could be used for
regularized models. For instance in a medical settings, we may want to regularize the
complexity of the treatment plan.

%We next discuss specific applications in EFG solving in detail.

% In standard sequential decision making the agent chooses a single action at
% every decision point, and receives a payoff either after selecting an action, or
% when reaching the end of the decision-making process. This process is modeled in
% our setting by setting $X_j=\Delta_{n_j}$ for all $j\in \convsets$, and having
% each loss function be linear.

% \todo{POMDPs???}
% \todo{Talk about finding sequential treatments cite IJCAI paper and Blue Skies paper and Mike Bowling's paper. }

%%% Local Variables:
%%% mode: latex
%%% TeX-master: "../aaai_2019/lrd_aaai_2019"
%%% End:

%% file: text/efg_special_case.tex
\subsection{Extensive-form games with convex-concave saddle-point structure}

In an extensive-form game with perfect recall each player faces a sequential decision-making problem,
of the type described in the previous section and in Figure~\ref{fig:kuhn treeplex player1}. The set of next potential
decision points $\childinfosets{j,a}$ is based on observations of stochastic
outcomes and actions taken by the other players.

%In an EFG setting the loss is typically a multilinear function describing the
%negative utility for the player, given the strategy of every other player at
%iteration $t-1$.

Here, we will focus on two-player zero-sum EFGs with perfect recall, but with
slightly more general utility structure than is usually considered. In
particular, we assume that we are solving a convex-concave saddle-point problem
of the following form:
\begin{align}
  \min_{x\in X} \max_{y \in Y} \big\{ \mu(x)^{\!\top}\! A \mu(y) + d_1(\mu(x)) - d_2(\mu(y)) \big\}, \label{eq:convex concave EFG}
\end{align}
where $X$ is the SDM polytope for Player $1$ and $Y$ is the SDM polytope of
Player $2$. Each $d_i$ is assumed to be a dilated convex function of the
form

\vspace{-4mm}
{\footnotesize
\begin{align*}
\vspace{-3mm}
  d_i(\mu(x)) = \sum_{j \in \convsets} \mu(x)_{p_j} d_j\bigg(\frac{\mu(x)_j}{\mu(x)_{p_j}}\bigg) = \sum_{j \in \convsets} \pi_j(x) \ell_j(x_j),
  \vspace{-4mm}
\end{align*}}
that is, in the form given in \eqref{eq:loss}.

% We denote $Q_X$ and $Q_Y$ the sequential action spaces of
% player 1 and 2, respectively.
In standard EFGs, the loss function for each player at each iteration $t$ is
 defined to be the negative payoff vector associated with the
sequence-form strategy of the other player at that iteration; since we
additionally allow a \emph{regularization} term we also get a nonlinear convex
term. More formally, at each iteration $t$, the loss functions
$\ell^t_X:X\to\mathbb{R}$ and
$\ell^t_Y:Y\to\mathbb{R}$ for player 1 and 2 respectively
are defined as

\vspace{-5mm}
{\footnotesize
\begin{align*}
  \ell^t_X : x \mapsto \langle -A\mu(y^{t}), \mu(x)\rangle + d_1(x), \\
  \ell^t_Y : y\mapsto \langle A^\top \mu(x^{t}), \mu(y)\rangle + d_2(y),
\end{align*}}
\vspace{-5mm}

\noindent where $A$ is the sequence-form payoff matrix of the
game~\citep{Stengel96:Efficient}. Some simple algebra shows that
$\ell^t_X$ and $\ell^t_Y$ are indeed separable (that is, they
can be written in the form of Equation~\ref{eq:loss}), where each
decision-point-level loss $\ell^t_{j, X}$ and $\ell^t_{j,Y}$
is a convex function.

This choice of loss function is justified by the fact that the induced
regret-minimizing dynamics for the two players lead to a convex-concave
saddle-point problem. Specifically, assume the two players play the game $T$
times, accumulating regret after each iteration as in Figure~\ref{fig:no
  alternation}.

\begin{figure}[ht]
\vspace{-4mm}
  \centering
  \begin{tikzpicture}[scale=0.85]
    \draw[thick] (0, 0) rectangle (1.2, .8);
    \node at (.6, .4) {$X$};
    \draw[thick] (0, -1) rectangle (1.2, -0.2);
    \node at (.6, -.6) {$Y$};
    \draw[->] (-.8, .4) -- (0, .4) node[above left] {$\ell_{X}^{t-1}$};
    \draw[->] (-.8, -.6) -- (0, -.6) node[above left] {$\ell_{Y}^{t-1}$};
    \draw[->] (1.2, .4) node[above right] {$x^t$} -- (2.0, .4);
    \draw[->] (1.2, -.6) node[above right] {$y^t$} -- (2.0, -.6);

    \draw[thick] (2.0, .2) rectangle (2.4, .6);
    \draw[thick] (2.0, .2) -- (2.4, .6);
    \draw[thick] (2.0, -.8) rectangle (2.4, -.4);
    \draw[thick] (2.0, -.8) -- (2.4, -.4);

    \draw[->] (2.4, .4) -- (2.6, .4) -- (3.2, -.6) -- (4, -.6) node[above left] {$\ell_{Y}^{t}$};
    \draw[->] (2.4, -.6) -- (2.6, -.6) -- (3.2, .4) -- (4, .4) node[above left] {$\ell_{X}^{t}$};

    \draw[thick] (4, 0) rectangle (5.2, .8);
    \node at (4.6, .4) {$X$};
    \draw[thick] (4, -1) rectangle (5.2, -0.2);
    \node at (4.6, -.6) {$Y$};
    \draw[->] (5.2, .4) node[above right] {$x^{t+1}$} -- (6.0, .4);
    \draw[->] (5.2, -.6) node[above right] {$y^{t+1}$} -- (6.0, -.6);

    \node at (6.5, -0.1) {$\cdots$};
    \node at (-1.3, -0.1) {$\cdots$};
  \end{tikzpicture}
  \vspace{-4mm}
  \caption{The flow of strategies and losses in regret minimization for games.
    The symbol \computeloss{} denotes computation/construction of the loss
    function.}
  \label{fig:no alternation}
\end{figure}
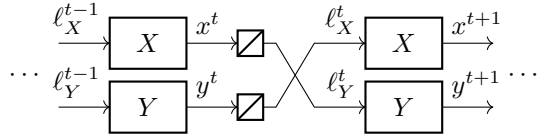

\noindent A folk theorem explains the tight connection between
low-regret strategies and approximate Nash equilibria. We will need a more
general variant of that theorem generalized to \eqref{eq:convex concave EFG}.
The convergence criterion we are interested in is the \emph{saddle-point
  residual (or gap)} $\xi$ of $(\bar x, \bar y)$, defined as
  {\footnotesize
\begin{equation*}
  \xi \!=\! \max_{\hat{y}} \{ d_1(\bar x) \!-\! d_2(\hat{y}) \!+\! \langle \bar x, A \hat y\rangle \}\! -\! \min_{\hat x} \{ d_1(\hat x) \!-\! d_2(\bar y) \!+\! \langle \hat x, A \bar y \rangle \}
\end{equation*}}
\vspace{-2mm}

We show that playing the average of a sequence of regret-minimizing
strategies leads to a bounded saddle-point residual. This result
is probably known, but it is unclear whether it has been stated in the form here.
We provide a proof in the appendix.
A closely related form is used for averaged strategy iterates in a first-order method
by \citet{Nemirovski04:Prox}.
\begin{restatable}{theorem}{folktheorem}\label{thm:folk theorem}
  If the average regret accumulated on $X$ and $Y$ by the two sets of strategies $\{x_t\}_{t=1}^T$ and $\{y_t\}_{t=1}^T$ is $\epsilon_1$ and $\epsilon_2$, respectively, then any strategy profile $(\bar x, \bar y)$ such that
  $
    \mu(\bar x) = \frac{1}{T}\sum_{t=1}^T \mu(x^t),\quad \mu(\bar y) = \frac{1}{T} \sum_{t=1}^T \mu(y^t)
  $
  has a saddle-point residual bounded by $\epsilon_1 + \epsilon_2$.
\end{restatable}

The above averaging is performed in the sequence-form space, which works because
that space is also convex. After averaging we can easily compute $\bar{x}$ in
linear time. Hence, by applying LRD to the decision spaces $X$ and
$Y$, we converge to a small saddle-point residual.

The fact that the averaging of the strategies is performed in sequence form explains why the traditional CFR presentation requires averaging with weights
based on the player's reaches $\pi_j$ at each decision point $j$.

Theorem~\ref{thm:folk theorem} shows that a Nash equilibrium can be computed by
taking the uniform distribution over sequence-form strategy iterates. However,
in the practical EFG-solving literature another approach called \emph{linear
  averaging} has been popular~\citet{Tammelin15:Solving}, especially for the
{\cfrp} algorithm. In linear averaging a weighted average strategy is
constructed, where each strategy $\mu(x^t)$ is weighted by $t$.
\citet{Tammelin15:Solving} show that this is guaranteed to converge specifically
when using the {\rmp} regret minimizer. It would be interesting to prove when
this works more generally. Here we make the simple observation that we can
compute \emph{both} averages, and simply use the one with better practical
performance, even in settings where only the uniform average is guaranteed to
converge.

%%% Local Variables:
%%% mode: latex
%%% TeX-master: "../aaai_2019/lrd_aaai_2019"
%%% End:

%% file: text/quantal_special_case.tex
\subsubsection{Quantal response equilibrium (QRE)}
\citet{Ling18:What} show that a reduced-normal-form QRE can be expressed as the
convex-concave saddle-point problem \eqref{eq:convex concave EFG} where $d_1$ and
$d_2$ are the (convex) dilated entropy functions usually used in first-order
methods (FOMs) for solving
EFGs~\citep{Hoda10:Smoothing,Kroer15:Faster,Kroer17:Theoretical}. This
saddle-point problem can be solved using FOMs, which would lead to fast
convergence rate due to the strongly convex nature of the dilated entropy
distance~\citep{Kroer17:Theoretical}. However, until now, no algorithms based on
local regret minimization at each decision point have been known for this
problem. Because the dilated entropy function separates into a sum over negative
entropy terms at each decision point it can be incorporated as a convex loss in
LRD. Combined with any regret-minimization algorithm that allows convex
functions over the simplex, this leads to the first regret-minimization algorithm
for computing (reduced-normal-form) QREs.
% Because we now have losses that consist
% of a strongly convex function and a linear term at each decision point we get an
% algorithm that is faster than CFR or LRD for computing Nash equilibria; for
% example, the strongly-convex variant of online gradient descent can be used to
% achieve an exponential convergence rate, rather than the usual rate of
% $T^{-\frac{1}{2}}$ for Nash equilibrium (see e.g. \citet{Hazan16:Introduction}
% for an overview of online gradient descent).

%%% Local Variables:
%%% mode: latex
%%% TeX-master: "../aaai_2019/lrd_aaai_2019"
%%% End:

% LocalWords:  LRD SDM

%% file: text/perturbed_efg_special_case.tex
\subsubsection{Perturbed EFGs and equilibrium refinement}
Equilibrium refinements are Nash equilibria with additional important rationality
properties. Such equilibria have rarely been used in practice due to scalability
issues. Recently, fast algorithms for computing approximate refinements were
introduced~\citet{Kroer17:Smoothing,Farina17:Regret}.
Theorem~\ref{thm:lrd} gives a new tool for constructing such methods: it
immediately implies correctness of the method of \citet{Farina17:Regret}, while
 also allowing new types of refinements and regret minimizers.

%%% Local Variables:
%%% mode: latex
%%% TeX-master: "../aaai_2019/lrd_aaai_2019"
%%% End:

%% file: text/cfr_plus.tex
\section{Erratum about Alternation in CFR$^+$}

Several tweaks to speed up the convergence of CFR have been proposed. The state of the art is CFR$^+$~\citep{Tammelin15:Solving}. %Here, we review
%the ideas in CFR$^+$, and discuss how they generalize to LRD.
{\cfrp}
consists of three tweaks: the {\rmp} regret minimizer, linear averaging, and
alternation. {\rmp} can be applied in our setting as well; it is simply an
alternative regret minimizer for linear losses over a simplex. We described linear averaging
earlier in this paper. Finally, {\em alternation} is the idea that at
iteration $t$, we provide Player $2$ with the utility vector associated with
the {\em current} iterate of Player $1$, rather than that of the previous
iteration, as is normally done in regret minimization.
Figure~\ref{fig:alternation} illustrates how this works, in contrast with
Figure~\ref{fig:no alternation} which shows the usual flow.
\begin{figure}[ht]
\vspace{-2mm}
  \centering
  \begin{tikzpicture}[scale=.85]
    \draw[thick] (0, 0) rectangle (1.2, .8);
    \node at (.6, .4) {${X}$};
    \draw[thick] (0, -1.2) rectangle (1.2, -0.4);
    \node at (.6, -.8) {${Y}$};
    \draw[->] (-.8, .4) -- (0, .4) node[above left] {$\ell_{\mathcal{X}}^{t-1}$};
    \draw[->] (1.2, .4) node[above right] {$x^t$} -- (2.0, .4);
    \draw[->] (1.2, -.8) node[below right] {$y^t$} -- (2.0, -.8);

    \draw[thick] (2.0, .2) rectangle (2.4, .6);
    \draw[thick] (2.0, .2) -- (2.4, .6);
    \draw[thick] (2.0, -1.0) rectangle (2.4, -.6);
    \draw[thick] (2.0, -1.0) -- (2.4, -.6);

    \draw[->] (2.4, .4) -- (2.6, .4) -- (2.6, -.2) -- (-.6, -.2) -- (-.6, -.8) -- (0, -.8) node[below left] {$\ell_{Y}^{t-1}$};
    \draw[->] (2.4, -.8) -- (2.6, -.8) -- (3.2, .4) -- (4, .4) node[above left] {$\ell_{X}^{t}$};

    \draw[thick] (4, 0) rectangle (5.2, .8);
    \node at (4.6, .4) {${X}$};
    \draw[thick] (4, -1.2) rectangle (5.2, -0.4);
    \node at (4.6, -.8) {${Y}$};
    \draw[->] (5.2, .4) node[above right] {$x^{t+1}$} -- (6.0, .4);

    \draw[thick] (6.0, .2) rectangle (6.4, .6);
    \draw[thick] (6.0, .2) -- (6.4, .6);

    \draw[->] (6.4, .4) -- (6.6, .4) -- (6.6, -.2) -- (3.4, -.2) -- (3.4, -.8) -- (4, -.8) node[below left] {$\ell_{Y}^{t}$};

    \draw[->] (5.2, -.8) node[below right] {$y^{t+1}$} -- (6.0, -.8);

    \draw[thick] (6.0, -1.0) rectangle (6.4, -.6);
    \draw[thick] (6.0, -1.0) -- (6.4, -.6);

    \draw (6.4, -.8) -- (6.6, -.8) -- (6.7, -.7);

    \node at (7.05, -0.1) {$\cdots$};
    \node at (-1.05, -0.1) {$\cdots$};
  \end{tikzpicture}
  \vspace{-3mm}
  \caption{The alternation method for CFR in games. The loss
     at iteration $t$ for $y$ is computed with $x^t$. The symbol
    \computeloss{} denotes computation/construction of the loss function.
    \vspace{-.2cm}}
  \label{fig:alternation}
\end{figure}
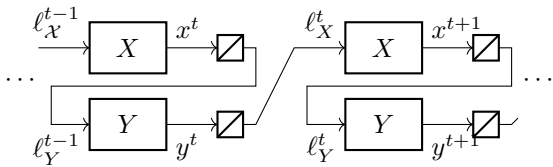
\citet{Tammelin15:Solving} state that they prove convergence of {\cfrp}; however
their proof relies on the folk theorem linking Nash equilibrium and regret. That
folk theorem is only proven for the case where no alternation is applied. We show below that the theorem does not hold with alternation!
\begin{observation}
  Let the action spaces for the players be $X = Y = [0,1]$, and let $\ell_X^t : x \mapsto x\cdot y^t$, $\ell_Y^t : y\mapsto -y\cdot x^{t+1}$ be bilinear loss functions (the superscript $t+1$ comes from the use of alternation---see Figure~\ref{fig:alternation}). Consider the sequence of strategies $x^t = t \mod 2, y^t = (t + 1) \mod 2$. A simple check reveals that after $2T$ iterations, the average regrets of the two players are both 0. Yet, the average strategies ${\bar x}^{2T} = {\bar y}^{2T} = 0.5$ do not converge to a saddle point of $xy$.
\end{observation}

This observation should be seen as more of a theoretical issue than practical;
alternation has been used extensively in practice, and the problem that we show
does not seem to come up for nondegenerate iterates (at least for {\cfrp}; it
may explain some erratic behavior that we have anecdotally observed with other
regret minimization algorithms when using alternation).
%In fact, it is possible
%that the folk theorem with alternation could be shown true under conditions that
%usually hold in practice, for example, that the uniform strategy is applied initially.
%Nonetheless, this does make {\cfrp} a heuristic algorithm as of now, rather than
%a sound algorithm.

%% file: text/experiments.tex
\section{Experiments}

We conducted multiple kinds of experiments on two EFG settings. The
first game is Leduc 5 poker~\citep{Southey05:Bayes}, a standard
benchmark in imperfect-information game solving. There is
a deck consisting of 5 unique cards with 2 copies of each. There are two rounds.
In the first round, each player places an ante of $1$ in the pot and receives a
single private card. A round of betting then takes place with a two-bet maximum,
with Player 1 going first. A public shared card is then dealt face up and
another round of betting takes place. Again, Player 1 goes first, and there is a
two-bet maximum. If one of the players has a pair with the public card, that
player wins. Otherwise, the player with the higher card wins. All bets in the first round are $1$, while all bets in the second round are $2$.
The second game is a variant of Goofspiel~\citep{Ross71:Goofspiel}, a bidding
game where each player has a hand of cards numbered $1$ to $N$. A third stack of
$N$ cards is shuffled and used as prizes: each turn a prize card is revealed,
and the players each choose a private card to bid on the prize, with the high
card winning, the value of the prize card is split evenly on a tie. After $N$
turns all prizes have been dealt out and the payoff to each player is the sum of
prize cards that they win. We use $N=4$ in our experiments.

First we investigate a setting where no previous regret-minimization algorithms
based on minimizing regret locally existed: the computation of QREs via LRM and
our more general convex losses. \citet{Ling18:What} use Newton's method for this
setting, but, as with standard Nash equilibrium, second-order algorithms do not
scale to large games (this is why {\cfrp} has been so
successful for creating human-level poker AIs). We compare
how quickly we can compute QREs compared to how quickly Nash equilibria can be
computed, in order to understand how large games we can expect to find QREs for
with our approach. To do this we run LRM with online gradient descent (OGD) at
each decision point. Because OGD is not guaranteed to stay within the simplex at
each iteration we need to project; this can be implemented via binary
search for decision points with large dimension~\citep{Duchi08:Efficient}, and
via a constant-size decision tree for low-dimension decision points. The results
are shown in Figure~\ref{fig:qre experiments}. We see that LRM performs
extremely well; in Goofspiel it converges vastly faster than {\cfrp}, and in
Leduc 5 it converges at a rate comparable to {\cfrp} and eventually becomes
faster. This shows that QRE computation via LRM likely scales to extremely large
EFGs, such as real-world-sized poker games (since {\cfrp} is known
to scale to such games).
\begin{figure}[t]
  \centering
  \includegraphics[width=0.85\columnwidth]{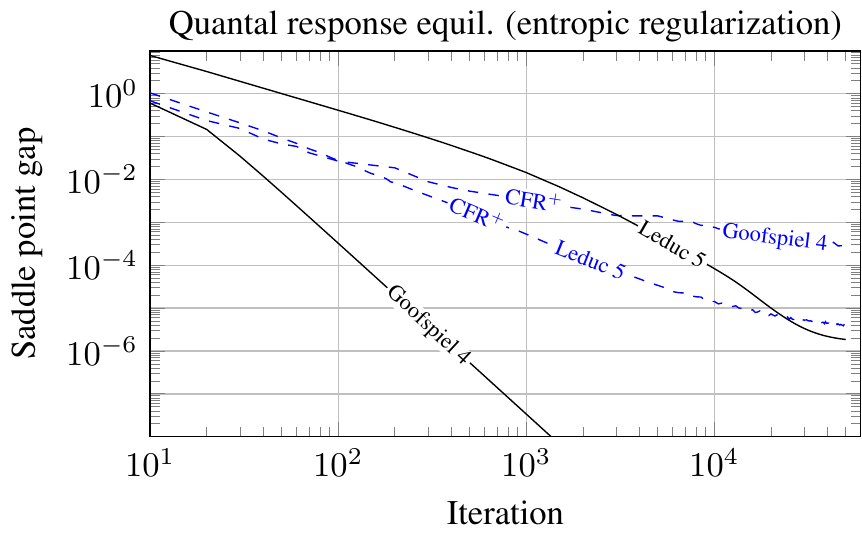}
  \vspace{-3mm}
  \caption{The QRE saddle-point gap as a function of the number of iterations
    for each game. The convergence rates of {\cfrp} for Nash equilibrium is
    shown for reference.\vspace{-3mm}}
  \label{fig:qre experiments}
\end{figure}

In the second set of experiments we investigate the speed of convergence for
solving $\ell_2$-regularized EFGs. Again we include the convergence rate of
standard CFR and {\cfrp} for Nash equilibrium computation as a benchmark. The
results for Leduc 5 are in Figure~\ref{fig:l2 leduc}. Solving
the regularized game is much faster than computing an Nash equilibrium via {\cfrp} except for
extremely small amounts of regularization. Results for Goofspiel
are in the appendix; the results are very
similar to the ones for Leduc 5.
\begin{figure}[th]
  \centering
  \includegraphics[width=0.85\columnwidth]{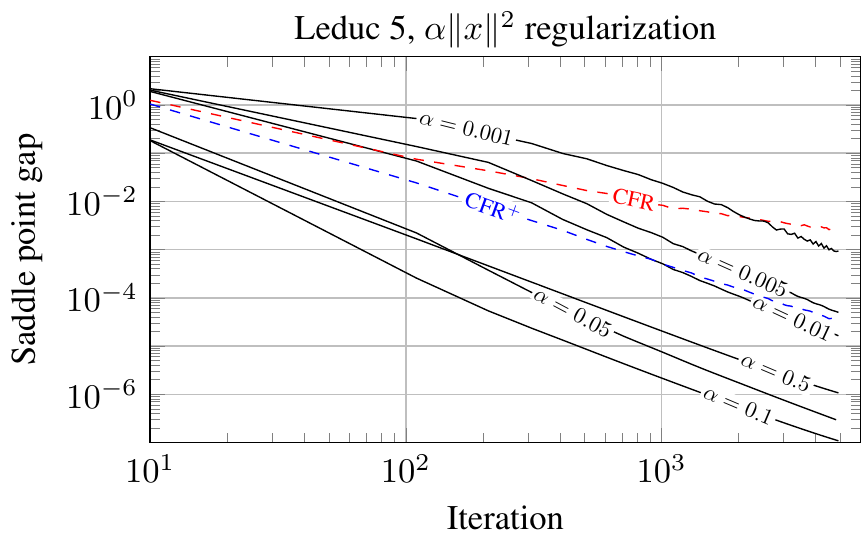}
  \vspace{-4mm}
  \caption{The saddle-point gap as a function of the number of iterations for
    $\ell_2$-regularized Leduc 5 for varying regularization amounts. The
    convergence rates of {\cfrp} for Nash equilibrium is shown for reference.\vspace{-4mm}}
  \label{fig:l2 leduc}
\end{figure}

In the third set of experiments we investigate the performance of LRM in
a single-agent-learning setting: learning how to exploit a static opponent where
we observe repeated samples from their strategy. We consider a
setting where the exploiter wishes to maximally exploit subject to staying near a pre-computed Nash equilibrium in
order to avoid opening herself up to exploitability~\citep{Ganzfried11:Game}. We model this in a new way: as a
regularized online SDM problem, where the loss for the exploiter is

\vspace{-5mm}
\begin{align}
  \ell^t(x) \defeq \langle -A\mu(y^t), \mu(x) \rangle + \alpha D(x \| x^{NE})
  \label{eq:nash distance regularized}
\end{align}
\vspace{-5mm}

\noindent where $y^t$ is the $t$'th observation from the opponent's strategy and $D(x \|
x^{NE})$ is the dilated $\ell_2$-based Bregman divergence between the NE
strategy $x^{NE}$ and $x$. The opponent's suboptimal strategy was computed by running CFR$^+$ until a gap of 0.1 was reached.
We stop the training of the exploiter after 5000 iterations or when an average regret of 0.0005 was reached, whichever happens first.
Figure~\ref{fig:exploit near
ne} shows the results. The ``utility increase'' line shows how much the agent gains by moving away
from the Nash equilibrium and towards an exploitative strategy, while the
``exploitability'' shows to what extent the agent thereby opens herself up to
being exploited by an optimal adversary. We see that this model can indeed be
used as a scalable proxy for trading off exploitation and exploitability by
varying $\alpha$.
\begin{figure}[th]
  \centering
  \includegraphics[width=0.85\columnwidth]{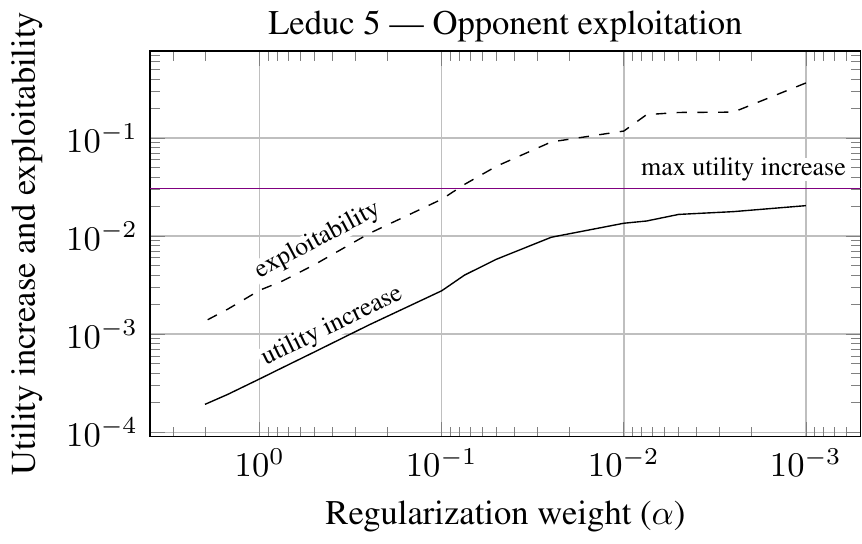}
  \vspace{-5mm}
  \caption{The utility increase from exploitation, and the resulting increase in
    the agent's own exploitability, as a function of decreasing penalization on
    distance from Nash equilibrium in \eqref{eq:nash distance regularized}. The
    straight line shows the value of best responding (i.e. maximally exploiting)
    to the opponent strategy.\vspace{-4mm}}
  \label{fig:exploit near ne}
\end{figure}

Our experiments are preliminary: we use simple OGD with very
little tuning, and focus on iteration complexity rather than runtime.
There are several reasons to expect that LRM for regularized games could be made
much faster. For one, the laminar losses are strongly convex,
so accelerated methods could be employed. This leads to significantly better
theoretical convergence rate than that of {\cfrp} for Nash equilibrium, and
could likely be exploited in practice also. Furthermore, we used OGD for
convenience, but one could most likely employ a projection-free algorithm and
thus have the computational cost at each decision point be the same as for
{\cfrp} while getting a faster convergence rate.
%%% Local Variables:
%%% mode: latex
%%% TeX-master: "../aaai_2019/lrd_aaai_2019"
%%% End:
% LocalWords:  OMD LRM LRD

%% file: text/conclusion.tex
\vspace{-1mm}\section{Conclusions and Future Research}

We presented LRD, a new decomposition of the regret associated with a sequential
action space into regrets associated with individual decision points. We developed
our technique for general compact convex sets and convex losses at each decision
point, thus providing a generalization of CFR beyond simplex decision points and
linear loss. We then showed that our results lead to a new class of
regret-minimization algorithms that solve SDM problems by minimizing regret
locally at each decision point. Although more general, our proof also provides a new
perspective on the CFR algorithm in terms of our regret decomposition, and we
explained the need for weighting by reach as a consequence of averaging in
sequence form. We then showed that our approach can be used to compute
regularized equilibria as well as Nash equilibrium refinements, and gave the
first regret-minimization algorithm for computing (reduced-normal-form) quantal response equilibrium (QRE) (based
on local regrets). We showed experimentally that even a preliminary variant of
LRM can be used to compute QREs in EFGs at a rate that is comparable to Nash equilibrium finding using {\cfrp},
thus yielding the first algorithm for computing QREs in extremely large EFGs. We
similarly showed that $\ell_2$-regularized equilibrium can be computed very quickly with out method.
Finally, we showed how our approach can be used as a new approach to opponent exploitation, and to control the tradeoff between exploitation and exploitability.

It would be interesting to investigate tweaks to the algorithms that use LRM in
order to understand what variants yield best practical
performance. It would also be interesting to find further applications where our
new types of decision spaces and loss functions can be used.
%, similar to the
%usage for computing QREs. Finally we showed that regularization in the form of
%\eqref{eq:nash distance regularized} can be used to trade off
%exploitability and exploitation; this could potentially lead to a viable method
%for controlled opponent exploitation in large-scale games.
%%% Local Variables:
%%% mode: latex
%%% TeX-master: "../aaai_2019/lrd_aaai_2019"
%%% End:

%% file: text/omitted_proofs.tex
\appendix
\section{Supplementary Material}
In this appendix we provide supplementary material that did not fit in the body of the paper.

\subsection{Omitted proofs}

\subsubsection{Folk theorem for convex-concave saddle point problems}

Here we will show a slightly more general version of the folk theorem, where
we allow convex losses. We must then evaluate the \emph{saddle-point residual}
$\xi$ of $(\bar x, \bar y)$, which is defined as:
\begin{align*}
  \xi &= \max_{\hat{y}} \{ d_1(\bar x) - d_2(\hat{y}) + \langle \bar x, A \hat y\rangle \} \\[-2mm]
      & \hspace{3.3cm}- \min_{\hat x} \{ d_1(\hat x) - d_2(\bar y) + \langle \hat x, A \bar y \rangle \}
\end{align*}

We now show that playing the average of a sequence of regret-minimizing
strategies leads to a bounded saddle-point residual. This result is also known,
though it is unclear whether it has been stated in the form used here. A very related
form is used for averaged strategy iterates in a first-order method by
\citet{Nemirovski04:Prox}.
\folktheorem*
\begin{proof}
  \begin{align*}
    R_1^T &\defeq\frac{1}{T}\sum_{t=1}^T \big(d_1(x^t) + \langle x^t, Ay^t\rangle\big)\\[-4mm]
    & \hspace{3.0cm}- \min_{\hat{x}} \left\{ d_1(\hat{x}) + \langle \hat{x}, A \bar{y}\rangle\right\} \le \epsilon_1
  \end{align*}
  \begin{align*}
    R_2^T &\defeq\frac{1}{T}\sum_{t=1}^T \big(d_2(y^t) - \langle x^t, Ay^t\rangle\big)\\[-4mm]
    & \hspace{3.0cm}- \min_{\hat{y}} \left\{ d_2(\hat{y}) - \langle \bar{x}, A\hat{y}\rangle\right\} \le \epsilon_2
  \end{align*}

  Now, we are interested in evaluating the gap $\xi$ of $(\bar x, \bar y)$. To this end:
  \begin{align*}
    \xi &= \max_{\hat{y}} \{ d_1(\bar x) - d_2(\hat{y}) + \langle \bar x, A \hat y\rangle \} \\[-2mm]
        & \hspace{3.3cm}- \min_{\hat x} \{ d_1(\hat x) - d_2(\bar y) + \langle \hat x, A \bar y \rangle \} \\
        &= \max_{\hat y} \{-d_2(\hat y) + \langle \bar x, A \hat y\rangle \} - \min_{\hat x} \{ d_1(\hat x) + \langle \hat x, A \bar y\rangle\}\\[-2mm]
        & \hspace{5.7cm} + d_1(\bar x) + d_2(\bar y)\\
        &= -\min_{\hat y} \{d_2(\hat y) - \langle \bar x, A \hat y\rangle \} - \min_{\hat x} \{ d_1(\hat x) + \langle \hat x, A \bar y\rangle\}\\[-2mm]
        & \hspace{5.7cm} + d_1(\bar x) + d_2(\bar y)
  \end{align*}
  By convexity of $d_1$ and $d_2$ we can write
  \[
    d_1(\bar x) + d_2(\bar y) \le \frac{1}{T} \sum_{t=1}^T d_1(x^t) + \frac{1}{T} \sum_{t=1}^T d_2(y^t)
  \]
  Substituting into the definition of $\xi$, we find that
  \begin{align*}
    \xi &\le \frac{1}{T}\sum_{t=1}^T d_2(y^t) -\min_{\hat y} \{d_2(\hat y) - \langle \bar x, A \hat y\rangle \}\\[-2mm]
        & \hspace{2.0cm} +\frac{1}{T}\sum_{t=1}^T d_1(x^t) -\min_{\hat x} \{ d_1(\hat x) + \langle \hat x, A \bar y\rangle\}\\
        &= R_1^T + R_2^T\\
        &\le \epsilon_1 + \epsilon_2,
  \end{align*}
  as we wanted to show.
\end{proof}

\thmlrdrm*
\begin{proof}
  From Theorem~\ref{thm:lrd} and by using the fact that
  \[-\min (f - g) \le -\min f +
  \max g \]
  we get that the cumulative regret can be written as the following
  recurrence:
  \begin{align*}
    & \cumulativeregret{j}{T} \le \laminarregret{j}{T}  +  \max_{\xhat_j \in X_j}\sum_{a\in A_j} \xhat_{j,a} \sum_{j' \in \childinfosets{j,a}} \!\!\cumulativeregret{j'}{T}.
  \end{align*}
  Applying this recurrence inductively gives the theorem.
\end{proof}

%% file: text/appendix_experiments.tex
\subsection{Additional experimental results}
\begin{figure}[h]
  \centering
  \includegraphics[width=\columnwidth]{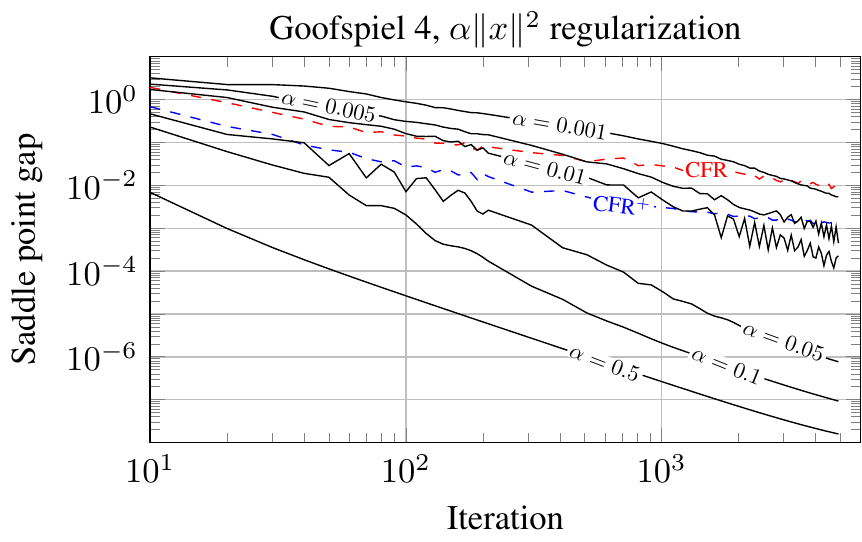}
  \caption{The saddle-point gap as a function of the number of iterations for
    $\ell_2$-regularized Goofspiel 4, for varying regularization amounts.}
  \label{fig:l2 goofspiel}
\end{figure}
%%% Local Variables:
%%% mode: latex
%%% TeX-master: "../aaai_2019/lrd_aaai_2019"
%%% End: